\documentclass[draft]{article}

\usepackage{amsmath, amssymb, amsthm, hyperref, mathtools}

\DeclareMathOperator{\chr}{char}
\DeclareMathOperator{\DFT}{DFT}

\newcommand{\CC}{\mathbb{C}}
\newcommand{\FF}{\mathbb{F}}
\newcommand{\QQ}{\mathbb{Q}}
\newcommand{\RR}{\mathbb{R}}
\newcommand{\ZZ}{\mathbb{Z}}
\newcommand{\NN}{\mathbb{N}}
\newcommand{\nN}{\mathcal{N}}

\newcommand{\AAA}{\mathbf{A}}
\newcommand{\aA}{\mathcal{A}}
\newcommand{\bB}{\mathcal{B}}
\newcommand{\ccC}{\mathcal{C}}
\newcommand{\dD}{\mathcal{D}}
\newcommand{\eE}{\mathcal{E}}

\newcommand{\prt}{\vee}

\newcommand{\brs}[1]{\left\{{#1}\right\}}

\newcommand{\prs}[1]{\left({#1}\right)}
\newcommand{\sPrs}[1]{\bigl({#1}\bigr)}
\newcommand{\sPRS}[1]{\biggl({#1}\biggr)}

\newcommand{\abs}[1]{\left\lvert{#1}\right\rvert}

\newcommand{\ceil}[1]{\left\lceil{#1}\right\rceil}
\newcommand{\sceil}[1]{\lceil{#1}\rceil}

\newcommand{\flor}[1]{\left\lfloor{#1}\right\rfloor}
\newcommand{\sflor}[1]{\lfloor{#1}\rfloor}
\newcommand{\sFlor}[1]{\bigl\lfloor{#1}\bigr\rfloor}

\newcommand{\cfrc}[2]{\ceil{\frac{#1}{#2}}}
\newcommand{\scfrc}[2]{\sceil{\frac{#1}{#2}}}

\newcommand{\ffrc}[2]{\flor{\frac{#1}{#2}}}
\newcommand{\sffrc}[2]{\sflor{\frac{#1}{#2}}}
\newcommand{\sFfrc}[2]{\sFlor{\frac{#1}{#2}}}

\newcounter{bibcnt}

\theoremstyle{definition}
\newtheorem{definition}{Definition}

\theoremstyle{plain}
\newtheorem{lemma}{Lemma}
\newtheorem{theorem}{Theorem}
\newtheorem*{conjecture}{Conjecture}
\newtheorem{corollary}{Corollary}

\theoremstyle{remark}
\newtheorem{remark}{Remark}

\begin{document}	

\title{Faster Polynomial Multiplication via Discrete Fourier Transforms}
\author{Alexey Pospelov\thanks{This research is supported by Cluster of Excellence ``Multimodal Computing and Interaction'' at Saarland University.}\\Computer Science Department, Saarland University\\\texttt{pospelov@cs.uni-saarland.de}}
\maketitle

\begin{abstract}
We study the complexity of polynomial multiplication over arbitrary fields. We present a unified approach that generalizes all known asymptotically fastest algorithms for this problem. In particular, the well-known algorithm for multiplication of polynomials over fields supporting DFTs of large smooth orders, Sch\"on\-hage-Stras\-sen's algorithm over arbitrary fields of characteristic different from 2, Sch\"onhage's algorithm over fields of characteristic 2, and Cantor-Kaltofen's algorithm over arbitrary algebras---all appear to be instances of this approach. We also obtain faster algorithms for polynomial multiplication over certain fields which do not support DFTs of large smooth orders.

We prove that the Sch\"onhage-Strassen's upper bound cannot be improved further over the field of rational numbers if we consider only algorithms based on consecutive applications of DFT, as all known fastest algorithms are. We also explore the ways to transfer the recent F\"urer's algorithm for integer multiplication to the problem of polynomial multiplication over arbitrary fields of positive characteristic.

This work is inspired by the recent improvement for the closely related problem of complexity of integer multiplication by F\"urer and its consequent modular arithmetic treatment due to De, Kurur, Saha, and Saptharishi. We explore the barriers in transferring the techniques for solutions of one problem to a solution of the other.
\end{abstract}

\section{Introduction}\label{sec:intro}

Complexity of polynomial multiplication is one of the central problems in computer algebra and algebraic complexity theory. Given two univariate polynomials by vectors of their coefficients,
\begin{align}\label{eq:ab}
a(x)&=\sum_{i=0}^{n-1} a_i x^i,&b(x)&=\sum_{j=0}^{n-1} b_jx^j,
\end{align}
over some field $k$, the goal is to compute the coefficients of their product
\begin{equation}\label{eq:c}
c(x)=a(x)\cdot b(x)=\sum_{\ell=0}^{2n-2} c_\ell x^\ell=\sum_{\ell=0}^{2n-2}\sum_{\substack{0\le i,\,j<n,\\i+j=\ell}}a_i b_j x^\ell.
\end{equation}
The direct way by the formulas above requires $n^2$ multiplications and $(n-1)^2$ additions of elements of $k$, making the total complexity of the naive algorithm $O(n^2)$. In what follows we call $k$ the \emph{ground field}.

\subsection{Model Of Computation}

We study the problem of the total \emph{algebraic} complexity of the multiplication of polynomials over \emph{fields}. That is, elements of $k$ are thought of as algebraic entities, and each binary arithmetic operation on these entities has unit cost. This model is rather abstract in the sense, that it counts, for example, an infinite precision multiplication of two reals as a unit cost operation. On the other hand, it has an advantage of being independent of any concrete implementation that may depend on many factors, including human-related, thus it is more universal, see the discussion on this topic in \cite[Introduction]{Burg}.

We are concerned with the \emph{total} number of arithmetic operations, i.e. multiplications and additions/subtractions that are sufficient to multiply two degree $n-1$ polynomials. Since the resulting functions can be computed without divisions, it seems natural to consider only \emph{division-free algebraic algorithms}. The inputs of such algorithm are the values $a_0,\,\dotsc,\,a_{n-1},\,b_0,\,\dotsc,\,b_{n-1}\in k$, the outputs are the values $c_0,\,c_1,\,\dotsc,\,c_{2n-2}\in k$ as defined in~\eqref{eq:ab}, \eqref{eq:c}. Any step of an algorithm is a multiplication, a division, an addition or a subtraction of two values, each being an input, a value, previously computed by the algorithm, or a constant from the ground field. An algorithm computes product of two degree $n-1$ polynomials, if all outputs $c_0,\,\dotsc,\,c_{2n-2}$ are computed in some of its steps. The number of steps of an algorithm $\aA$ is called \emph{algebraic} or \emph{arithmetic} complexity of $\aA$.

In what follows, we will always consider division-free algebraic algorithms. A multiplication performed in a step of an algorithm is called \emph{scalar}, if at least one multiplicand is a field constant, and \emph{nonscalar} in the other case. For an algorithm $\aA$ which computes the product of two degree $n-1$ polynomials, we define $L^m_\aA(n)$ to be the number of nonscalar multiplications used in $\aA$, and $L^a_\aA(n)$ to be the total number of additions, subtractions and scalar multiplications in $\aA$. We also set $L_\aA(n)\coloneqq L_\aA^m(n)+L_\aA^a(n)$, the \emph{total algebraic complexity} of $\aA$ computing the product of two degree $n-1$ polynomials. In what follows, $\AAA_k^n$ always stands for the set of division-free algorithms computing the product of two degree $n-1$ polynomials over $k$,
\begin{align*}
L^m_k(n)&\coloneqq\min_{\aA\in\AAA_k^n}L^m_\aA(n),&
L^a_k(n)&\coloneqq\min_{\aA\in\AAA_k^n}L^a_\aA(n),&
L_k(n)&\coloneqq\min_{\aA\in\AAA_k^n}L_\aA(n).
\end{align*}
When the field $k$ will be clear from the context or insignificant, we will use then the simplified notation: $L^m(n)$, $L^a(n)$ and $L(n)$, respectively. Note, that $L(n)$ needs not to be equal to $L^m(n)+L^a(n)$, since the minimal number of nonscalar multiplications and the minimal number of additive operations and scalar multiplications can be achieved by different algorithms.

\subsection{Fast Polynomial Multiplication And Lower Bounds}

Design of efficient algorithms and proving lower bounds is a classical problem in algebraic complexity theory that received wide attention in the past. For an exhaustive treatment of the current state of the art we advise the reader to refer to~\cite[Sections~2.1, 2.2, 2,7, 2.8]{Burg}. There exists an algorithm $\aA\in\AAA_k^n$, such that
\begin{align}\label{eq:nlogn}
L_\aA^m(n)&=O(n),&L_\aA^a(n)&=O(n\log n),&L_\aA(n)&=O(n\log n),\footnotemark{}
\end{align}%
\footnotetext{In this paper we always use $\log\coloneqq\log_2$.}%
if $k$ supports Discrete Fourier Transformation (DFT) of order $2^l$,~\cite[Chapter~1, Section~2.1]{Burg} or $3^l$,~\cite[Exercise~2.5]{Burg} for each $l>0$. Sch\"onhage-Strassen's algorithm $\bB\in\AAA_k^n$ computes the product of two degree $n-1$ polynomials over an arbitrary field $k$ of characteristic different from $2$ with
\begin{equation}\label{eq:nlognloglogn}
\begin{split}
L_\bB^m(n)=O(n\log n),&\qquad L_\bB^a(n)=O(n\log n\log\log n),\\
L_\bB(n)&=O(n\log n\log\log n).
\end{split}
\end{equation}
cf.~\cite{SS71}, \cite[Section 2.2]{Burg}. In fact, the original algorithm of~\cite{SS71} computes product of two $n$-bit integers, but it readily transforms into an algorithm for degree $n-1$ polynomial multiplication. For fields of characteristic $2$, Sch\"onhage's algorithm~\cite{Sc77}, \cite[Exercise~2.6]{Burg} has the same upper bounds as in~\eqref{eq:nlognloglogn}. An algorithm $\ccC'$ for multiplication of polynomials over arbitrary rings with the same upper bound for $L^m_{\ccC'}(n)$ was first proposed by Kaminski in~\cite{Ka88}. However, there was no matching upper bound for $L^a_{\ccC'}(n)$. Cantor and Kaltofen generalized Sch\"onhage-Strassen's algorithm into an algorithm $\ccC$ for the problem of multiplication of polynomials over arbitrary algebras (not necessarily commutative, not necessarily associative) achieving the upper bounds~\eqref{eq:nlognloglogn}, see~\cite{Caka}.

For the rest of the paper, we will use the introduced notation: $\aA$ will always stand for the multiplication algorithm via DFT with complexity upper bounds~\eqref{eq:nlogn}, $\bB$ will stand for Sch\"onhage-Strassen's algorithm if $\chr k\neq 2$ and for Sch\"onhage's algorithm if $\chr k=2$, both with complexity upper bounds~\eqref{eq:nlognloglogn}, and $\ccC$ will stand for Cantor-Kaltofen's algorithm for multiplication of polynomials over arbitrary algebras with the same complexity upper bounds as Sch\"onhage-Strassen's algorithm.

Upper and lower bounds for $L^m_k(n)$, which is also called the \emph{multiplicative complexity}, received special attention in literature, see, e.g., \cite[Section~14.5]{Burg}. It is interesting, that for each $k$, there exists always an algorithm $\eE\in\AAA_k^n$ with $L^m_\eE(n)=O(n)$, if we do not worry that $L^a_\eE(n)$ will be worse than in \eqref{eq:nlognloglogn}, see~\cite{Chud,Shpr}.

If $\abs{k}\ge 2n-2$, then it is known, that $L^m(n)=2n-1$, see~\cite[Theorem~(2.2)]{Burg}. For the fields $k$ with $n-2\le\abs{k}\le 2n-3$, the exact value for $L^m_k(n)=3n-\sFfrc{\abs{k}}{2}-2$ was proved by Kaminski and Bshouty in~\cite[Theorem~2]{KB89} (see~\cite[Lemma~1]{BK90} for the proof of the theorem to hold for the multiplicative complexity).

In order to multiply two degree $n-1$ polynomials over $\FF_q$ it suffices to pick an irreducible over $\FF_q$ polynomial $p(x)$ of degree $2n-1$ and multiply two elements in $\FF_q[x]/p(x)$, that is in $\FF_{q^{2n-1}}$. Therefore, for finite fields $k=\FF_q$ with $\abs{k}=q\le n-3$, currently best upper bounds for $L^m_{\FF_q}(n)$ are derived from Chudnovskys' algorithm for multiplication in finite field extensions~\cite{Chud,Shpr} and its improvements by Ballet et al. ($p$ stands always for a prime number; in fact all of the following upper bounds hold also for the \emph{bilinear} complexity, which is a special case of multiplicative complexity, when each nonscalar multiplication in an algorithm is of kind $\ell(a_0,\,\dotsc,\,a_{n-1})\cdot\ell'(b_0,\,\dotsc,\,b_{n-1})$ for some linear forms $\ell,\,\ell'\in(k^n)^\ast$):
$$
L^m_{\FF_q}(n)\le\begin{cases}
4(1+\frac{1}{\sqrt{q}-3})n+o(n),&q=p^{2\kappa}\ge 25,\text{ \cite[Theorem~7.7]{Chud}},\\
4(1+\frac{p}{\sqrt{q}-3})n,&q=p^{2\kappa}\ge 16,\text{ \cite[Theorem~3.1]{Ba03}},\\
6(1+\frac{4}{q-3})n,&q=p\ge 5,\text{ \cite[Theorem~2.3]{BC04}},\\
6(1+\frac{2p}{q-3})n,&q=p^\kappa\ge 16,\text{ \cite[Theorem~4.6]{BBR09}},\\
12(1+\frac{p}{q-3})n,&q>3,\text{ \cite[Corollary~3.1]{Ba03}},\\
54n-27,&q=3,\text{ \cite[Remark after Corollary~3.1]{Ba03}},\\
\frac{477}{13}n-\frac{108}{13}<36.7n,&q=2,\text{ \cite[Theorem~3.4]{BP10}}.
\end{cases}
$$

The best known lower bounds in case of $k=\FF_q$ when $q\le n-3$ are
$$
L_{\FF_q}(n)\ge L^m_{\FF_q}(n)\ge\begin{cases}
\prs{3+\frac{(q-1)^2}{q^5+(q-1)^3}}n-o(n),&q\ge 3,\text{ \cite{Ka05}},\\
3.52n-o(n),&q=2,\text{ \cite{BD80}}.
\end{cases}
$$
If we allow for a moment divisions to be present in an algorithm, then there is a lower bound $3n-o(n)$ for the total number of nonscalar multiplications \emph{and} divisions necessary for any algebraic algorithm computing product of two degree $n$ polynomials, see~\cite{BK06}.

There are few lower bounds for the algebraic complexity of polynomial multiplication. Most of them are actually bounding $L^m(n)$ which can be used as a conservative lower bound for $L(n)$. Since the coefficients $c_0,\,\dotsc,\,c_{2n-2}$ are linearly independent, in case of division-free algorithms one immediately obtains the lower bound $L(n)\ge L^m(n)\ge 2n-1$ over arbitrary fields. To the moment, this is the only general lower bound for $L(n)$ which does not depend on the ground field. B\"urgisser and Lotz in~\cite{BL04} proved the only currently known nonlinear lower bound if $\Omega(n\log n)$ for $L_\CC(n)$ (actually, on $L^a_\CC(n)$) which holds in case when all scalar multiplications in an algorithm are with bounded constants.

The gap between the upper and the lower bounds on $L_k(n)$ motivates to look for better multiplication algorithms and for higher lower bounds for the complexity of polynomial multiplication, in particular over small fields. For example, it is still an open problem if the total algebraic complexity of polynomial multiplication is nonlinear, see \cite[Problem~2.1]{Burg}. Another well known challenge is to decrease the upper bound for $L_k(n)$ of~\eqref{eq:nlognloglogn} to the level of~\eqref{eq:nlogn} in case of arbitrary fields, see~\cite{Pan94} for the more general challenge of multivariate polynomial multiplication. In this paper we partially address both problems.

\subsection{Our Results}

As our first contribution, for every field $k$, we present an algorithm $\dD_k\in\AAA_k^n$, which is a generalization of Sch\"onhage-Strassen's construction that works over arbitrary fields and achieves the best known complexity upper bounds. In fact, we argue that the algorithm $\dD_k$ stands for a generic polynomial multiplication algorithm that relies on consecutive application of DFT. In particular, the algorithms $\aA$, $\bB$, and $\ccC$ come as special cases of the algorithm $\dD_k$. We are currently not aware of any algorithms with an upper bound of~\eqref{eq:nlognloglogn} that are not based on consecutive DFT applications and thus do not follow from the algorithm $\dD_k$.

As the second contribution, we show that $L_{\dD_k}(n)=o(n\log n\log\log n)$ in case when algorithm $\aA$ cannot be applied but the field $k$ has some simple algebraic properties that are ignored by algorithms $\bB$ and $\ccC$. This improves the upper bound of~\eqref{eq:nlognloglogn} over such fields. We also present a parameterization of fields $k$ with respect to the performance of the algorithm $\dD_k$, and give explicit upper bounds which depend on this parameterization. More precisely, over each field $k$, we have $\Omega(n\log n)=L_{\dD_k}(n)=O(n\log n\log\log n)$, and over certain fields that do not admit low-overhead application of the algorithm $\aA$, the algorithm $\dD_k$ achieves intermediate complexities between the indicated bounds.

Finally, we show, that the algorithm $\dD_k$ has natural limitations depending on the ground field $k$. For example, we prove that $L_{\dD_\QQ}(n)=\Omega(n\log n\log\log n)$. Furthermore, we characterize all such fields, where application of DFT-based methods does not lead to any improvement of the upper bound~\eqref{eq:nlognloglogn}. Therefore, we consider this as an exhaustive exploration of performance of generic algorithms for polynomial multiplication based on application of DFT.

\subsection{Organization Of the Paper}

Section~\ref{sec:defs} contains the necessary algebraic preliminaries. We then give a uniform treatment of the best known algorithms for polynomial multiplication over arbitrary fields in Section~\ref{sec:state}: Sch\"onhage-Strassen's algorithm~\cite{SS71}, Sch\"onhage's algorithm~\cite{Sc77} and Cantor-Kaltofen's algorithm~\cite{Caka}. In Section~\ref{sec:fft} we remind the best known upper bounds for computation of DFT over different fields and show some efficient applications of their combination. We also indicate limitations of the known techniques.

Section~\ref{sec:approach} contains our main contributions. We end with one particular number-theoretic conjecture due to Bl\"aser on the existence of special finite field extensions. In fact, if it holds, then the algorithm algorithm $\dD_k$ can achieve better performance than that of the previously known algorithms $\bB$ and $\ccC$ over any field of characteristic different from $0$.

\section{Basic Definitions}\label{sec:defs}

In what follows we will denote the ground field by $k$. \emph{Algebra} will always stand for a finite dimensional associative algebra over some field with unity $1$. For a function $f:\NN\to\RR$, a positive integer $n$ is called \emph{$f$-smooth}, if each prime divisor of $n$ does not exceed $f(n)$. Note, that this definition is not trivial only if $f(n)<\frac{n}{2}$. If $f(n)=O(1)$, then an $f$-smooth positive integer is called just \emph{smooth}.

All currently known fastest algorithms for polynomial multiplication over arbitrary fields rely on the possibility to apply the Discrete Fourier Transform by means of the Fast Fourier Transform algorithm (FFT) and on the estimation of the overhead needed to extend the field to make DFTs available. This possibility depends on existence of so-called \emph{principal roots of unity} of large smooth orders, e.g., of orders $2^\nu$ for all $\nu>0$.

Let $A$ be an algebra over a field $k$. $\omega\in A$ is called a \emph{principal $n$-th root of unity} if $\omega^n=1_A$ (where $1_A$ is the unity of $A$) and for $1\le\nu<n$, $1-\omega^\nu$ is not a zero divisor in $A$. It follows, that if $\omega\in A$ is a principal $n$-th root of unity, then $\chr k\nmid n$ and
\begin{equation}\label{eq:prin}
\sum_{\nu=0}^{n-1}\omega^{i\cdot\nu}=\begin{cases}
n,&\text{if }i\equiv 0\pmod{n},\\
0,&\text{otherwise.}
\end{cases}
\end{equation}
If $A$ is a field, then $\omega\in A$ is a principal $n$-th root of unity iff $\omega$ is a primitive $n$-th root of unity. For a principal $n$-th root of unity $\omega\in A$, the map
$$
\DFT_n^\omega:\:A[x]/(x^n-1)\to A^n
$$
defined as $\DFT_n^\omega\prs{\sum_{\nu=0}^{n-1} a_\nu x^\nu}=(\tilde a_0,\,\dotsc,\,\tilde a_{n-1})$, where $\tilde a_i=\sum_{\nu=0}^{n-1}\omega^{i\cdot\nu} a_\nu$, for $i=0,\,\dotsc,\,n-1$, is called the \emph{Discrete Fourier Transform} of order $n$ over $A$ with respect to the principal $n$-th root of unity $\omega$.

It follows from Chinese Remainder Theorem that if $\omega\in A$ is a principal $n$-th root of unity, then $\DFT_n^\omega$ is an isomorphism between $A[x]/(x^n-1)$ and $A^n$. \eqref{eq:prin} implies that the inverse transform of $\DFT_n^\omega$ is $\frac{1}{n}\DFT_n^{\omega^{-1}}$ since $\omega^{-1}$ is also a principal $n$-th root of unity in $A$~\cite[Theorem~(2.6)]{Burg}: $a_i=\frac{1}{n}\sum_{\nu=0}^{n-1} \omega^{-i\cdot\nu} \cdot \tilde{a}_\nu$, for $i=0,\,\dotsc,\,n-1$. Note, that if $\omega\in A$ is a principal $n$-th root of unity and $a(x)=a_0+a_1 x+\dotsb+a_{n-1} x^{n-1}\in k[x]/(x^n-1)$, then
$$
\DFT_n^\omega\prs{a(x)}=\prs{a(\omega^0),\,a(\omega),\,\dotsc,\,a(\omega^{n-1})}.
$$

An important property of the DFT is that it can be computed efficiently under certain conditions, see Section~\ref{sec:fft}. We only mention here, that if $n=s^\nu$ for some constant $s$, there is a principal $n$-th root of unity $\omega$ in an algebra $A$, then $\DFT_n^\omega$ can be computed in $O(n\log n)$ additions of elements of $A$ and multiplications of elements of $A$ with powers of $\omega$.

\section{State Of the Art}\label{sec:state}

\subsection{Multiplication via DFT}

The easiest way to illustrate power of applications of DFT is to consider multiplication of polynomials over a field $k$ which contains primitive roots of unity of large smooth orders. Assume that for some integer constant $s\ge 2$ and for each $\nu$, $k$ contains a primitive \mbox{$s^\nu$-th} root of unity. The well-known DFT-based algorithm $\aA$ takes two degree $n-1$ polynomials $a(x)$ and $b(x)$ and proceeds as follows:
\begin{description}
\item[Embed and pad] Set $\nu=\ceil{\log_s(2n-1)}$ such that $s^\nu\ge 2n-1$. Pad the vectors of coefficients of $a(x)$ and $b(x)$ with zeroes and consider $a(x)$ and $b(x)$ as polynomials of degree $s^\nu-1$ in $k[x]/(x^{s^\nu}-1)$. This step is performed at no arithmetical cost.
\item[Compute DFTs] For a primitive $s^\nu$-th root of unity $\omega\in k$, compute
\begin{align*}
\tilde{a}&\coloneqq\DFT_{s^\nu}^\omega(a(x)),&\tilde{b}&\coloneqq\DFT_{s^\nu}^\omega(b(x)).
\end{align*}
The cost of this step is $O(n\log n)$ arithmetical operations over $k$ (recall, that $s$ is a constant).
\item[Multiply vectors] Compute dot-product $\tilde c\coloneqq\tilde a\cdot\tilde b$, that is perform $s^\nu=O(n)$ multiplications of elements in $k$.
\item[Compute inverse DFT] Compute
$$
\frac{1}{s^\nu}\DFT_{s^\nu}^{\omega^{-1}}(\tilde c)=c(x).
$$
This step requires $O(n\log n)$ arithmetical operations in $k$.
\end{description}

As we can see the total complexity of $O(n\log n)$ arithmetic operations over $k$. Note, that the number of multiplications is $s^\nu\le 2ns-s$, and is linear in $n$ as long as $s$ is a constant.

\subsection{Multiplication in Arbitrary Fields}

Now suppose that $k$ does not contain the needed primitive roots of unity. The methods we will describe now are all based on the idea of an algebraic extension $K\supset k$ where the DFT of a large smooth order $s^\nu$ is defined. In these methods one encodes the input polynomials into polynomials of smaller degree over $K$ and uses the algorithm $\aA$ over $K$ to multiply these polynomials. The $s^\nu$ multiplications of elements in $K$ are performed via an efficient reduction to multiplication of polynomials of smaller degree, thus making the whole scheme recursive.

\subsubsection{Sch\"onhage-Strassen's Algorithm}

Assume that $\chr k\neq 2$. In this case, $x$ is a $2n$-th principal root of unity in $A_n:=k[x]/(x^n+1)$, which is a $k$-algebra of dimension $n\dim A$~\cite[(2.11)]{Burg} and $A[x]/(x^n+1)\cong A[x]/(x^n-1)$, if a $k$-algebra $A$ contains a principal \mbox{$2n$-th} root of unity~\cite[(2.12)]{Burg}. For $n\ge 3$, Sch\"onhage-Strassen's algorithm~\cite{SS71}, which we denote by $\bB$ takes two degree $n-1$ polynomials $a(x)$ and $b(x)$ over $k$ and proceeds as follows:
\begin{description}
\item[Embed and pad] Set $\nu=\ceil{\log_2(2n-1)}\ge 2$ such that $N\coloneqq 2^\nu\ge 2n-1$. Pad the vectors of coefficients of $a(x)$ and $b(x)$ with zeroes and consider $a(x)$ and $b(x)$ as polynomials of degree $N-1$ in $A_N$. This step is performed at no arithmetical cost.
\item[Extend] Set $N_1\coloneqq 2^{\cfrc{\nu}{2}}\ge 2$, $N_2\coloneqq 2^{\ffrc{\nu}{2}+1}$, such that $\frac{N_1}{2}\cdot N_2=N$. Encode $a(x)$ and $b(x)$ (considered as elements of $A_N$) as polynomials of degree $N_2-1$ over $A_{N_1}=k[y]/(y^{N_1}+1)$:
$$
a(x)=\sum_{i=0}^{N-1}a_ix^i\mapsto\sum_{i=0}^{N_2-1}\underbrace{\prs{\sum_{j=0}^{\frac{N_1}{2}-1}a_{\frac{N_1}{2}\cdot i+j}y^j}}_{\eqqcolon \bar a_i\in A_{N_1}}{\underbrace{(x^{\frac{N_1}{2}})}_{\bar x}}^i=\sum_{i=0}^{N_2-1}\bar a_i\bar x^i\eqqcolon\bar a(\bar x).
$$
$y$ is a $2N_1$-th principal root of unity in $A_{N_1}$ and $2N_1\ge N_2$, all powers of $2$. Since $N_2\mid 2N_1$, $\psi\coloneqq y^{\frac{2N_1}{N_2}}$ is a principal $N_2$-th root of unity in $A_{N_1}$.
\item[Compute DFTs] of orders $N_2$ of $\bar a(\bar x)$ and $\bar b(\bar x)$ with respect to $\psi$. Note, that addition of two elements in $A_{N_1}$ can be performed in $N_1$ additions in $A$, and multiplication by powers of $\psi$, that is, by powers of $y$ results in cyclic shifts and sign changes and is also bounded by $N_1$ additions (if we count a sign change as an additive operation). Therefore, this step requires $O(N_1\cdot N_2\log N_2)=O(N\log N)$ arithmetic operations over $k$.
\item[Multiply] the coordinates of $\tilde{\bar a}\cdot\tilde{\bar b}=\tilde{\bar c}$. This results in computing $N_2$ products of polynomials of degree $\frac{N_1}{2}-1$, which are computed by a recursive application of the currently described procedure.
\item[Compute inverse DFT] of $\tilde{\bar c}$ with respect to $\psi^{-1}=y^{2N_1-\frac{2N_1}{N_2}}$. As before, this requires $O(N\log N)$ additive operations in $k$.
\item[Unembedding] in this case is can be computed in the following way: since degrees in $y$ of all coefficients $\bar a_i$, $\bar b_i$ were at most $\frac{N_1}{2}-1$, and they were multiplied in $A_{N_1}$, degrees in $y$ of all coefficients $\bar c_i$ are at most $N_1-2<N_1$. Therefore, for all $i=0,\,\dotsc,\,N_2-1$,
$$
\bar c_i=\sum_{j=0}^{N_1-1}c_{i,\,j}y^j
$$
are already computed with some $c_{i,\,j}\in k$, and
\begin{multline*}
c(x)=\sum_{i=0}^{N_2-1}\bar c_i (x^{\frac{N_1}{2}})^i=\sum_{i=0}^{N_2-1}\sum_{j=0}^{N_1-1} c_{i,\,j} x^{\frac{N_1}{2}\cdot i+j}\\
=\sum_{i=0}^{N-1}(c_{\sffrc{2i}{N_1},\,i-\sffrc{2i}{N_1}\cdot\frac{N_1}{2}}+c_{\sffrc{2i}{N_1}-1,\,\frac{N_1}{2}+i-\sffrc{2i}{N_1}\cdot\frac{N_1}{2}})x^i
\end{multline*}
can be computed by at most $N$ additions of elements in $k$ (we assume that $c_{i,\,j}=0$ if $i<0$ or $j\ge N_1$).
\end{description}
Denoting by $L'_\bB(N)$ the total complexity of multiplication in $A_{N}$ via Sch\"onhage-Strassen's algorithm $\bB$, we obtain following complexity inequality:
$$
L_\bB(n)\le L'_\bB(N)\le N_2 L'_\bB\prs{N_1}+O(N\log N).
$$
It implies $L'_\bB(N)=O(N\log N\log\log N)$ and the desired estimates~\eqref{eq:nlognloglogn} since $N\le 4n-2$. A more careful examination of the numbers of additions and multiplications used gives also the upper bounds~\eqref{eq:nlognloglogn}.

Rough complexity analysis can be also made by following observations. The cost of each recursive step (under a recursive step we understand all the work done on a fixed recursive depth) is $O(N_1\cdot N_2\log N_2)=O(n\log n)$ and is defined by the complexity of the DFT used to reduce the multiplication to several multiplications of smaller formats. Note, that in order to adjoin a $2N_1$-th root of unity to $k$ in the initial step we take a (ring) extension of degree $N_1$, which is a half of the degree of the root we get. This crucial fact reduces the number of recursive steps to $O(\log\log n)$. Thus, the upper bounds~\eqref{eq:nlognloglogn} for the complexity of $\bB$ can also be obtained as a product of the upper bound for the complexity of a recursive step by the number of recursive steps.

\subsubsection{Sch\"onhage's Algorithm}

Now assume that $\chr k=2$. Again, the first step is the choice of a finite dimensional algebra to reduce the original polynomial multiplication to. In case of $\chr k=2$, the choice of $k[x]/(x^n+1)$ does not work since it can be used only efficient to append $2^\nu$-th roots of unity and $x^{2^\nu}-1=(x-1)^{2^\nu}$ in every field of characteristic $2$. Sch\"onhage's algorithm~\cite{Sc77} thus reduces the multiplication of polynomials over $k$ to the multiplication in $B_N\coloneqq k[x]/(x^{2N}+x^N+1)$, where $x$ is a $3N$-th principal root of unity. Therefore, we can follow the way of the original Sch\"onhage-Strassen's algorithm with one important modification explained in this section.

For $n\ge 3$, Sch\"onhage's algorithm  $\bB$ takes two degree $n-1$ polynomials $a(x)$ and $b(x)$ and proceeds as follows:
\begin{description}
\item[Embed and pad] Set $\nu=\ceil{\log_3(n-\frac{1}{2})}$ such that for $N\coloneqq 3^\nu$, $2N\ge 2n-1$. Pad the vectors of coefficients of $a(x)$ and $b(x)$ with zeroes and consider $a(x)$ and $b(x)$ as elements of $B_N$. This step is performed at no arithmetical cost.
\item[Extend] Set $N_1\coloneqq 3^{\cfrc{\nu}{2}}$ and $N_2\coloneqq 3^{\ffrc{\nu}{2}}$ such that $N_1N_2=N$. Encode the input polynomials $a(x)$ and $b(x)$ (considered as elements of $B_N$) as polynomials of degree $2N_2-1$ over $B_{N_1}=k[y]/(y^{2N_1}+y^{N_1}+1)$:
$$
a(x)=\sum_{i=0}^{2N-1}a_i x^i\mapsto\sum_{i=0}^{2N_2-1}\underbrace{\prs{\sum_{j=0}^{N_1-1}a_{N_1\cdot i+j}y^j}}_{\eqqcolon \bar a_i\in B_{N_1}}{\underbrace{(x^{N_1})}_{\bar x}}^i=\sum_{i=0}^{2N_2-1}\bar a_i\bar x^i\eqqcolon\bar a(\bar x).
$$
$y$ is a $3N_1$-th principal root of unity in $B_{N_1}$, and $N_1\ge N_2$, both powers of $3$. Thus, $\psi=y^{\frac{N_1}{N_2}}$ is a $3N_2$-th principal root of unity in $B_{N_1}$.
\item[Compute DFTs] of $\bar a(\bar x)$ and $\bar b(\bar x)$, both padded to degree $3N_2$ with zeroes, with respect to $\psi$. Note, that addition of two elements in $B_{N_1}$ can be performed in at most $2N_1$ additions of elements in $k$, and multiplications by powers of $\psi$, that is, by powers of $y$ can also be performed in $O(N_1)$ operations since $y^{3N_1i+\ell}=y^\ell$, $y^{3N_1i+2N_1+\ell'}=-y^{N_1+\ell'}-y^{\ell'}$ for every $i\ge 0$, $0\le\ell<2N_1$, and $0\le\ell'<N_1$. Therefore, multiplication of any element of $B_{N_1}$ by a power of $y$ can be performed by at most one addition of two polynomials in $B_{N_1}$ and sign inversion of it, that is, in at most $4N_1$ additive operations in $k$ (again, if we count a sign inversion as an operation with unit cost, otherwise it is just $2N_1$). Overall, this step requires $O(N_1\cdot N_2\log N_2)=O(N\log N)$ operations in $k$.
\item[Multiply] component-wise two vectors of length $3N_2$, $\tilde{\bar a}$ and $\tilde{\bar b}$. Note, however, that only $2N_2$ out of these products are enough, namely only $\tilde{\bar a}_i\cdot\tilde{\bar b}_i$ where $i\not\equiv 0\pmod{3}$. This is explained in the next step.
\item[Compute inverse DFT] of $(\bar{\tilde c}_0,\,\dotsc,\,\bar{\tilde c}_{3N_2-1})$ in $O(N\log N)$ operations in $k$. This computes the coefficients of $\bar c'(\bar x)=\bar a(\bar x)\bar b(\bar x)\pmod{{\bar x}^{3N_2}-1}$, and we need
$$
\bar c(\bar x)=\bar a(\bar x)\bar b(\bar x)\pmod{x^{2N_2}+x^{N_2}+1}.
$$
This is resolved by noticing that
\begin{align*}
\bar c_i&=\bar c'_i-\bar c'_{i+2N_2},&\bar c_{i+N_2}&=\bar c'_{i+N_2}-\bar c'_{i+2N_2},
\end{align*}
for all $i=0,\,\dotsc,\,N_2-1$. To compute these differences, consider the explicit formulas of the direct DFT of order $3N_2$ with respect to $\psi$:
$$
\tilde{\bar c}_{3i+j}=\sum_{\nu=0}^{3N_2-1}\bar c'_\nu\psi^{3i\nu+j\nu}=\sum_{\nu=0}^{N_2-1}\bar c'_{\nu,\,j}\psi^{3i\nu}\eqqcolon\tilde{\bar c}_{i,\,j},
$$
\begin{equation}\label{eq:schon}
\bar c'_{i,\,j}=\frac{1}{N_2}\sum_{\nu=0}^{N_2-1}\tilde{\bar c}_{\nu,\,j}\psi^{-3i\nu},
\end{equation}
for $0\le i<N_2$, $0\le j\le 2$ and $\bar c'_{i,\,j}=\frac{1}{3}(\bar c'_i+\psi^{jN_2}\bar c'_{i+N_2}+\psi^{2jN_2}\bar c'_{i+2N_2})\cdot\psi^{ij}$. Therefore,
$$
\bar c'_{i+jN_2}=\frac{1}{3}(\bar c'_{i,\,0}+\psi^{-2jN_2-i}\bar c'_{i,\,1}+\psi^{-jN_2-2i}\bar c'_{i,\,2})
$$
and the required differences
\begin{align*}
\bar c'_i-\bar c'_{i+2N_2}&=\frac{1}{3}\sPrs{(\psi^{-i}-\psi^{-N_2-i})\bar c'_{i,\,1}+(\psi^{-2i}-\psi^{-2N_2-2i})\bar c'_{i,\,2}},\\
\bar c'_{i+N_2}-\bar c'_{i+2N_2}&=\frac{1}{3}\sPrs{(\psi^{-2N_2-i}-\psi^{-N_2-i})\bar c'_{i,\,1}+(\psi^{-2i}-\psi^{-N_2-2i})\bar c'_{i,\,2}},
\end{align*}
can be computed from $\bar c'_{i,\,j}$ for $j=1,\,2$, which can be computed via~\eqref{eq:schon} from $\tilde{\bar c}_{i,\,j}=\tilde{\bar c}_{3i+j}=\tilde{\bar a}_{3i+j}\tilde{\bar b}_{3i+j}$ for $i=0,\,\dotsc,\,N_2-1$ and $j=1,\,2$, that is from $2N_2$ products.
\item[Unembed] in the similar way as in the original Sch\"onhage-Strassen's algorithm. This requires $O(N)$ operations in $k$.
\end{description}

If we denote again $L'_\bB(N)$ the total complexity of multiplication in $B_{N}$ via Sch\"onhage's algorithm $\bB$, we obtain following complexity inequality:
$$
L_\bB(n)\le L'_\bB(N)\le 2N_2 L'_\bB\prs{N_1}+O(N\log N).
$$
It implies $L'_\bB(N)=O(N\log N\log\log N)$ and the desired estimates~\eqref{eq:nlognloglogn} since $N\le 3n-2$. Again, a more careful examination of the numbers of additions and multiplications used again gives also the upper bounds~\eqref{eq:nlognloglogn}.

\subsubsection{Cantor-Kaltofen's Generalization}

In~\cite{Caka} Cantor and Kaltofen presented a generalized version of Sch\"on\-ha\-ge-Stras\-sen's algorithm~\cite{SS71}, an algorithm $\ccC$ which computes the coefficients of a product of two polynomials over an arbitrary, not necessarily commutative, not necessarily associative algebra with unity with upper bounds~\eqref{eq:nlognloglogn}. Here we present a simplified version of this algorithm which works over fields, or, more generally, over division algebras. We will use this restriction to perform divisions by constants of an algebra via multiplication by inverses of these constants.

Let $\omega\in\ccC$ be a primitive $n$-th root of unity. Then $\Phi_n(x)=\prod_{(i,\,n)=1}(x-\omega^i)$ is called a \emph{cyclotomic polynomial} of order $n$. One easily deduces that for each $n$,
$$
\Phi_n(x)\mid(x^n-1)=\prod_{0\le i<n}(x-\omega^i).
$$
It is well known, that all coefficients of $\Phi_n(x)$ are integers, for every $n$, $\Phi_n(x)$ is irreducible over $\QQ$, and or any $s,\,n$, $\Phi_{s^n}(x)=\Phi_s(x^{n-1})$. The degree of $\Phi_n(x)$ is the number of natural numbers $i\le n$, coprime with $n$, which is denoted by $\phi(n)$ and called \emph{Euler's totient function}. Trivially, $\phi(n)\le n-1$ with an equality iff $n$ is a prime, and for $n\ge 3$, $\phi(n)>\frac{1}{2}\cdot\frac{n}{\log n}$, see~\cite{HS69}. From the above properties of $\Phi_n(x)$ we also have $\phi(s^n)=s^{n-1}\phi(s)$ for all $s,\,n\ge 1$. Therefore, if $s$ is a constant and $n$ grows, then the number of monomials in $\Phi_{s^n}(x)$ is bounded by a constant (for example, $s-1$).

Let $k$ be a field of characteristic $p$, and $s\ge 2$ be some integer, that will be fixed throughout of the entire algorithm, $p\nmid s$. Cantor-Kaltofen's algorithm takes two degree $n-1$ polynomials $a(x)$ and $b(x)$ over $k$ for $n\ge s^{3}$ and proceeds as follows:
\begin{description}
\item[Embed and pad] Set $\nu\coloneqq\ceil{\log_s\sPrs{(4n-2)\log s}}$, such that
$$
N\coloneqq\phi(s^\nu)\ge 2n-1.
$$
The multiplication is then performed in $C_N\coloneqq k[x]/\Phi_N(x)$, where $x$ is a principal $N$-th root of unity.
\item[Extend] Set $N_1\coloneqq s^{\sffrc{\nu}{2}}\phi(s)$, $N_2\coloneqq s^{\scfrc{\nu}{2}-1}$ such that for
\begin{align*}
N_3&\coloneqq s^{\sffrc{\nu}{2}+1},&N_1&=\phi(N_3)\ge N_2,&N_1N_2&=N.
\end{align*}
Note, that $sN_2\mid N_3$. Encode polynomials $a(x)$ and $b(x)$ (considered as elements of $C_N$) as polynomials of degree $N_2-1$ over $C_{N_3}$:
$$
a(x)=\sum_{i=0}^{N-1}a_i x^i\mapsto\sum_{i=0}^{N_2-1}\underbrace{\prs{\sum_{j=0}^{N_1-1}a_{i+N_2j}y^j}}_{\eqqcolon\bar a_i\in C_{N_3}}{\bar x}^i\eqqcolon\bar a(\bar x).
$$
$y$ is a principal $N_3$-th root of unity in $C_{N_3}$, therefore, $\psi=y^{\frac{N_3}{N_2}}$ is a principal $N_2$-th root of unity and $\xi=y^{\frac{N_3}{sN_2}}$ is a principal $sN_2$-th root of unity in $C_{N_3}$.

Note, that $x^{N_1}\mapsto y$, and the polynomials $\bar a_i=\bar a_i(y)$ are in fact of degree at most $\scfrc{N_1-1}{2}$. This follows from the fact, that $a_l=0$ for $l\ge n$, that is, for $i+N_2j\ge n$, for $0\le i<N_2$ and $0\le j<N_1$. One can easily verify, that it is equivalent to the inequality $j\le\frac{n}{N_2}-1\le\frac{N-1}{2N_2}-1\le\sffrc{N_1}{2}-1$. Therefore, multiplication of any two polynomials taken from the linear span of $\bar a_i$ modulo $\Phi_{N_3}(y)$ is in fact the ordinary multiplication of these polynomials.
\item[Compute DFTs] $\tilde{\bar a}=\DFT_{N_2}^\psi(\bar a(\bar x))$, $\tilde{\bar a}'=\DFT_{N_2}^\psi(\bar a(\xi\bar x))$, $\tilde{\bar b}=\DFT_{N_2}^\psi(\bar b(\bar x))$, and $\tilde{\bar b}'=\DFT_{N_2}^\psi(b(\xi\bar x))$. Precomputation of coefficient of $a(\xi x)$ and $b(\xi x)$ requires $O(N_2)$ multiplications by small powers of $y$ in $C_{N_3}$. Computation of the DFTs requires $O(N_2\log N_2)$ additions and multiplications by powers of $\psi$, that is, by powers of $y$, in $C_{N_3}$. Note, that, as usual, addition of two elements in $C_{N_3}$ requires $N_2=\phi(N_3)$ additions of elements in $k$.

Multiplications by powers of $\psi$, that is, by powers of $y$, can be first performed modulo $X^{N_3}-1$ at no cost (since they are in this case simply cyclic shifts), and then by reduction modulo $\Phi_{N_3}(x)$. This is possible since $\Phi_{N_3}(x)$ divides $x^{N_3}-1$. Since $\Phi_{N_3}(x)$ is monic and has at most $s$ nonzero monomials, such a reduction can be performed with at most $(s-1)(N_3-N_2)=O(N_3)$ scalar multiplications and the same number of additions of elements in $k$. Therefore, the total cost of this step is $O(N_3\cdot N_1\log N_1)=O(N\log N)$ since $N_3\le\frac{s}{\phi(s)}N_2\le2\log s\cdot N_2=O(N_2)$ and $N_1N_2=N$.
\item[Multiply] component-wise two pairs of vectors of length $N_2$: $\tilde{\bar c}''=\tilde{\bar a}\cdot\tilde{\bar b}$ and $\tilde{\bar c}'=\tilde{\bar a}'\cdot\tilde{\bar b}'$. This is performed recursively by the same procedure since the components of these vectors are elements in $C_{N_3}$.
\item[Compute inverse DFTs] $\bar c'=\DFT_{N_2}^{\psi^{-1}}(\tilde{\bar c}')$ and $\bar c''=\DFT_{N_2}^{\psi^{-1}}(\tilde{\bar c}'')$. This requires again $O(N\log N)$ steps, as in the computation of the direct DFTs.

Now recall, that we need the coefficients $\bar c_i\in C_{N_3}$ of the product of polynomials $\bar c(x)=\bar a(x)\bar b(x)\pmod{\Phi_N(x)}$. For this, we shall first compute the coefficients $\hat c_0,\,\dotsc,\,\hat c_{2N_2-2}$ of the regular polynomial product $\hat c(x)=\bar a(x)\bar b(x)$. These can easily be computed from the $\bar c'_i$, $\bar c''_i$ via the following formulas for $0\le i<N_2$:
\begin{align*}
\hat c_i&=\frac{1}{N_2(1-\xi^{N_2})}(\bar c_i''-\xi^{N_2}\bar c_i'),&\hat c_{N_2+i}&=\frac{1}{N_2(1-\xi^{N_2})}(\bar c_i'-\bar c_i'').
\end{align*}
In order to get rid of divisions in $C_{N_3}$ we can use the identity
$$
\frac{1}{1-\xi^{N_2}}=\frac{1}{\tau}\prod_{\substack{2\le i<s,\\(i,\,s)=1}}(1-\xi^{N_2i}),
$$
where $\tau=1$ if $s$ is not a prime power, and $\tau=p$ if $s=p^\kappa$ for some prime $p$. Note, that in the latter case necessarily $\chr k\neq p$. This identity shows how one can compute the fraction $\frac{1}{1-\xi}$ in $2\phi(s)-1$ additions and multiplications by powers of $y$ in $C_{N_3}$ without divisions: multiplication of the intermediate product $\Pi$ by the next factor $1-\xi^{N_2i}$ can be computed as $\Pi-\xi^{N_2i}\Pi$. Therefore, all coefficients $\hat c_i$ for $0\le i\le 2N_2-2$ can be computed in $O(N)$ operations in $k$. In order to obtain the coefficients of $\bar c(x)$, it suffices to reduce the polynomial $\hat c(x)$ modulo $\Phi_{N_3}(x)$ which can be performed in $O(N)$ steps, as explained before.
\item[Unembedding] in this case is not needed because of the choice of the encoding of polynomials: coefficients $\bar c_i$ computed in the Multiplication step, decoded back by substituting $y\mapsto x^{N_2}$, turn into polynomials in $x$ with monomials of pairwise different degrees for different $i=0,\,\dotsc,\,N-1$.
\end{description}

If we denote $L'_\ccC(N)$ the total complexity of multiplication in $C_{N}$ via Cantor-Kaltofen's algorithm $\ccC$, we obtain following complexity inequality:
$$
L_\ccC(n)\le L'_\ccC(N)\le 2N_2 L'_\ccC\prs{N_1}+O(N\log N).
$$
The choice of parameters $N_1$ and $N_2$ implies $L'_\ccC(N)=O(N\log N\log\log N)$ and the desired estimates~\eqref{eq:nlognloglogn} since $N<(s-1)\log s\cdot (2n-1)=O(n)$. A more careful examination of the numbers of additions and multiplications used again gives also the upper bounds~\eqref{eq:nlognloglogn}.

If $\chr k\neq 2$ and $s=2$, then $\Phi_{s^\nu}(x)=x^{2^{\nu-1}}+1$, and we get the multiplication in the algebra $A_{2^{\nu-1}}$ from the Sch\"onhage-Strassen's algorithm. If $\chr k\neq 3$ and $s=3$, then $\Phi_{s^\nu}(x)=x^{2\cdot 3^{\nu-1}}+x^{3^{\nu-1}}+1$ and we get the multiplication in the algebra $B_{3^{\nu-1}}$. However, the multiplication is performed differently: instead of performing one DFT of order $N_2\sim 2\sqrt{N}$ over $A_{N_1}$ (of order $3N_2$ over $B_{N_1}$ with only $2N_2\sim 2\sqrt{N}$ multiplications sufficient, resp.), Cantor-Kaltofen's algorithm performs two DFTs of order $N_2\sim\sqrt{N}$ over $C_{N_3}$.

Summarizing the above algorithms of complexity $O(n\log n\log\log n)$ we notice that in case, when it is impossible to apply FFT directly in the ground field, a ring extension is always introduced. Since the costs of all recursive steps are roughly the same, total complexity of such an algorithm can be naturally bounded by the product of the cost of one recursive step by the number of steps, which is $O(\log\log n)$ in the algorithm $\bB$. Complexity of one recursive step is defined by the complexity of computing DFTs, for which nothing better than $O(n\log n)$-time algorithms for computing of a DFT of order $n$ is currently known. The first potential improvement of this scheme is to reduce the complexity of algorithms computing DFT. The second is reducing the number of recursive steps of such an algorithm. In the first case we can increase the number of recursive steps needed, depending on the boost we will achieve in computing DFT. In the second case we can increase the number of operations used by DFT computations, however, we must always make sure that the product of these two values does not exceed $\Omega(n\log n\log\log n)$. In this paper we are concerned mostly with the problem of reduction of the recursive depth of such algorithms. Effectivity of our solution appears to depend only on algebraic properties of the ground field.

\section{An Upper Bound for the Complexity of DFT}\label{sec:fft}

In this section we summarize the best known upper bounds for the computation of DFTs over an algebra $A$ with unity $1$. Let $\omega\in A$ be a principal $n$-th root of unity. For $a(x)\in A[x]$ of degree $n-1$ let $\tilde a=\DFT_n^\omega(a(x))\in A^n$. We will denote the total number of operations over $A$ that are sufficient for an algebraic algorithm to compute the DFT of order $n$ over $A$ by $D_A(n)$. In case, when the algebra $A$ be insignificant or clear from the context, we will use the notation $D(n)$.

There is always an obvious way to compute $\tilde a$ from the coefficients of $a(x)$.

\begin{lemma}\label{lem:triv}
For every $A$ and $n\ge 1$, such that the DFT of order $n$ is defined over $A$,
\begin{equation}\label{eq:dfttriv}
D_A(n)\le\begin{cases}
2n^2-3n+1,&\text{if }2\nmid n,\\
2n^2-5n+4,&\text{if }2\mid n.
\end{cases}
\end{equation}
\end{lemma}

\begin{proof}
To compute $\tilde a_0$, $n-1$ additions are always sufficient. Let $\omega\in A$ be a principal $n$-th root of unity. If $2\mid n$, then $\omega^{\frac{n}{2}}=-1$, and to compute $\tilde a_{\frac{n}{2}}$, $n-1$ additions/subtractions are also sufficient. For the rest of the coefficients $\tilde a_i$, one always needs $n-1$ additions and, in case of odd $n$, $n-1$ multiplications by powers of $\omega$. For even $n$, one multiplication can be saved, namely, by $\omega^{i\frac{n}{2}}=(-1)^i$, it can be implemented by selective changing the sign of the corresponding additive operation in the sum for $\tilde a_i$. Therefore, we obtain
$$
D_A(n)\le\begin{cases}
(n-1)+2(n-1)^2=2n^2-3n+1,&\text{if }2\nmid n,\\
2(n-1)+(n-2)((n-2)+(n-1))=2n^2-5n+4,&\text{if }2\mid n,
\end{cases}
$$
which proves the statement.
\end{proof}

The next method of effective reduction of a DFT of large order to DFTs of smaller orders is known as Cooley-Tukey's algorithm~\cite{Cool}, \cite[Section~4.1]{Clau} and is based on the following lemma which directly follows from the well-known facts and is present here for completeness.

\begin{lemma}\label{lem:ct}
Let the DFT of order
\begin{equation}\label{eq:n}
n=p_1^{d_1}\dotsc p_s^{d_s}\ge 2
\end{equation}
be defined over $A$ \textup($p_\sigma$ are not necessary prime and even pairwise coprime\textup). Then
\begin{equation}\label{eq:coltuk}
D(n)\le n\sum_{\sigma=1}^s\prs{\frac{d_\sigma}{p_\sigma}(D(p_\sigma)-1)+d_\sigma}-n+1.
\end{equation}
\end{lemma}
\begin{proof}
We first prove that if $n=n_1n_2$, then
\begin{equation}\label{eq:cool}
D(n)\le n_1D(n_2)+n_2D(n_1)+(n_1-1)(n_2-1).
\end{equation}
Let $\omega\in A$ be a principal $n$-th root of unity. Then $\omega_1\coloneqq\omega^{n_2}$ is a principal $n_1$-th root of unity and $\omega_2\coloneqq\omega^{n_1}$ is a principal $n_2$-th root of unity. For a polynomial $a(x)\in A[x]/(x^n-1)$, consider $\tilde a=\DFT_n^\omega(a(x))$: for $0\le j<n_2$, $0\le l<n_1$
\begin{multline*}
\tilde a_{n_1j+l}=\sum_{\nu=0}^{n-1}a_\nu\omega^{(n_1j+l)\nu}=\sum_{\nu=0}^{n_1-1}\sum_{\mu=0}^{n_2-1}a_{n_2\nu+\mu}\omega^{n_2\nu(n_1j+l)+\mu(n_1j+l)}\\
=\sum_{\mu=0}^{n_2-1}\biggl(\omega^{\mu l}\underbrace{\sum_{\nu=0}^{n_1-1}a_{n_2\nu+\mu}\omega_1^{\nu l}}_{\eqqcolon\tilde a_{\mu,\,l}}\biggr)\omega_2^{\mu j}=\sum_{\mu=0}^{n_2-1}\underbrace{(\omega^{\mu l}\tilde a_{\mu,\,l})}_{\eqqcolon\hat a_{\mu,\,l}}\omega_2^{\mu j}
=\sum_{\mu=0}^{n_2-1}\hat a_{\mu,\,l}\omega_2^{\mu j}\eqqcolon\tilde a_{j,\,l}.
\end{multline*}
Computation of all values $\tilde a_{j,\,l}$ for a fixed $l$ can be performed via the DFT of order $n_2$ with respect to $\omega_2$. Therefore, to compute all values $\tilde a_{j,\,l}$, i.e., all values $a_i$ for $0\le i<n$, it suffices to perform $n_1$ DFTs of order $n_2$. Computation of all values $\tilde a_{\mu,\,l}$ for fixed $\mu$ can be performed via the DFT of order $n_1$ with respect to $\omega_1$. Therefore, to compute all values $\tilde a_{\mu,\,l}$, it suffices to perform $n_2$ DFTs of order $n_1$. Finally, to compute $\hat a_{\mu,\,l}$ from $\tilde a_{\mu,\,l}$, one needs one multiplication by $\omega^{\mu l}$ if $\mu>0$ and $l>0$ (if $\mu=0$ or $l=0$ then no computation is needed). This takes $(n_1-1)(n_2-1)$ multiplications by powers of $\omega$ to compute all values $\hat a_{\mu,\,l}$. This proves~\eqref{eq:cool}.

\eqref{eq:coltuk} follows by consecutive application of~\eqref{eq:cool} choosing $d_1$ times $p_1$ for $n_1$, then $d_2$ times $p_2$, etc. Noting that $D(1)=0$ completes the proof.
\end{proof}
\begin{corollary}\label{cor:coltuk}
Let $n$ be as in~\eqref{eq:n}, and let all $2=p_1<p_2<\dotsb<p_s$ be all primes. Then
\begin{equation}\label{eq:pfft}
D(n)\le\sPRS{\frac{3}{2}d_1+2\sum_{\sigma=2}^s{d_\sigma(p_\sigma-1)}-1}n+1.
\end{equation}
In particular,
\begin{equation}\label{eq:pfftu}
D(n)\le 2\max_{1\le\sigma\le s}p_\sigma\cdot n\log n.
\end{equation}
\end{corollary}
\begin{proof}
\eqref{eq:pfft} follows from~\eqref{eq:coltuk} by applying the upper bound of Lemma~\ref{lem:triv} for the values of $D(p_\sigma)$.

Obviously $d_1,\,\dotsc,\,d_s\le\log n$ since $p_\sigma^{d_\sigma}\le n$, $p_\sigma\ge 2$ for $1\le\sigma\le s$. Therefore, $$
D(n)\le\prs{\frac{3}{2}+2\sPrs{\max_{1\le\sigma\le s}p_\sigma-1}-1}n\log n+1\le 2\max_{1\le\sigma\le s}p_\sigma\cdot n\log n,
$$
which proves~\eqref{eq:pfftu}.
\end{proof}

Lemma~\ref{lem:ct} provides an efficient method of reduction of a DFT of composite order $n$ to several DFTs of smaller orders which divide $n$. For example, if all $p_\sigma$ in \eqref{eq:n} are bounded by some constant, then~\eqref{eq:pfftu} shows that Cooley-Tukey's algorithm computes the DFT of order $n$ in $O(n\log n)$ steps. Furthermore, if $\max_{1\le\sigma\le s} p_\sigma\le g(n)$ for some slowly growing function $g(n)$, say $g(n)=o(\log\log n)$, then~\eqref{eq:pfftu} gives an upper bound of $o(n\log n\cdot g(n))$ for the computation of the DFT of order $n$. However, this method fails to be effective if $n$ has large prime factors (or is just prime). We could use the algorithm from Lemma~\ref{lem:triv}, but sometimes we can apply Rader's algorithm to compute a DFT of prime order~\cite{Rader}, \cite[Section 4.2]{Clau}.

\begin{lemma}\label{lem:rader}
Let $p$ be a prime, and assume that the DFT of order $p$ is defined over $A$.
\begin{enumerate}
\item If the DFT of order $p-1$ is defined over $A$, then $D(p)\le 2D(p-1)+O(p)$.
\item If for $n>2p-4$, the DFT of order $n$ is defined over $A$, then
$$
D(p)\le 2D(n)+O(n).
$$
\end{enumerate}
\end{lemma}
\begin{remark}
Note, that the first bound can be efficient if $p-1$ is a smooth number. Otherwise we may choose some larger smooth $n$ for the second case, making sure that the DFT of order $n$ exists over $A$ and $n$ is not too large, that is, to achieve an $O(p\log p)$ upper bound for $D(p)$.
\end{remark}
\proof
Let $\omega\in A$ be a principal $p$-th root of unity. For a polynomial
$$
a(x)\in A[x]/(x^p-1),
$$
the value of $\tilde a_0=\sum_{i=0}^{p-1}a_i$ can be computed directly by performing $p-1$ additions. For $1\le i\le p-1$,
\begin{equation}\label{eq:radmainstep}
\tilde a_i-a_0=\sum_{j=1}^{p-1} a_j\omega^{ij}\eqqcolon\tilde a_i'.
\end{equation}
Thus, to compute all $\tilde a_i$ from $\tilde a_i'$, $p-1$ additions are enough.

\begin{enumerate}
\item The multiplicative group $\FF_p^\ast=\brs{1\le i<p}$ is isomorphic to the cyclic group $\ZZ_{p-1}$ with $p-1$ elements. We will denote the isomorphism by $\alpha$. For $a''_{i-1}\coloneqq a_{\alpha(i)}$ and $\tilde a_{i-1}''=\tilde a_{\alpha(i)}'$, from \eqref{eq:radmainstep} we obtain
$$
\tilde a_i''=\sum_{j=1}^{p-1} a_j\omega^{\alpha(i)+\alpha(j)}=\sum_{j=0}^{p-2} a''_j\omega^{\alpha(i+j)}.
$$
The latter is a cyclic convolution, which can be performed via computing the coefficients of the product of the degree $p-2$ polynomial
$$
a''(x)=\sum_{i=0}^{p-2}a''_i x^i,
$$
and the degree $p-2$ polynomial with fixed coefficients
$$
\omega(x)=\sum_{i=0}^{p-2}\omega^{\alpha(i)} x^i.
$$
This can be achieved by computing the DFT of $a''(x)$, performing $p-1$ multiplications by constants (components of the DFT of $\omega(x)$, in fact, these are just polynomials in $\omega$), and computing the reverse DFT. This proves the first bound.
\item For an $n\ge 2p-3$, we may define the polynomials
\begin{align*}
\hat a(x)&=a''_0+a''_1x^{n-p+2}+\dotsb+a''_{p-2}x^{n-1},&\hat\omega(x)&=\sum_{i=0}^{n-1}\omega^{\alpha(i\bmod(p-2)+1)}x^i
\end{align*}
and compute their cyclic convolution. Then the first $p-1$ coefficients of the cyclic convolution will be exactly the $a''_0,\,\dotsc,\,a''_{p-2}$. Note, that again, we do not need to count the complexity of the DFT of $\hat\omega(x)$ since it is fixed and can be precomputed. This proves the second bound.\qed\end{enumerate}

\begin{corollary}\label{cor:dft}
Let $p$ be a fixed odd prime, $k$ be a field where the DFT of order $p^N-1$ is defined for $N=2^n$, $n\ge\ceil{\log(2p-5)}$. Then $D_k(p^N-1)=O(p^{N}\cdot N^2)$.
\end{corollary}
\begin{proof}
We have $p^N-1=(p-1)(p+1)(p^2+1)\dotsm(p^{2^{n-1}}+1)$. Since $p$ is odd, each factor is even and $p^N-1=2^n\cdot\frac{p-1}{2}\prod_{i=1}^{n-1}\frac{p^{2^i}+1}{2}$. Let $p^N-1=p_1^{d_1}\dotsm p_s^{d_s}$ be the decomposition of $p^N-1$ into primes and $p_1=2<p_2<\dotsb<p_s$, and $p_2,\,\dotsc,\,p_{i_1}$ are all less than $\frac{p-1}{2}$, $p_{i_1+1},\,\dotsc,\,p_{i_2}$ are less than $\frac{p+1}{2}$, and, in general, $p_{i_j+1},\,\dotsc,\,p_{i_{j+1}}$ are less or equal than $\frac{p^{2^{j-1}}+1}{2}$. Note, that $i_n=s$. We also set $i_{-1}=0,\,i_0=1$. From~\eqref{eq:coltuk} we have
$$
D(p^N-1)\le (p^N-1)\sum_{\sigma=1}^s\prs{\frac{d_\sigma}{p_\sigma}(D(p_\sigma)-1)+d_\sigma}-p^N+2.
$$
Obviously, for $p_1=2$, we have $D(p_1)=2\le p_1\cdot\log p_1$. Using Lemma~\ref{lem:rader} we can compute the DFT of orders $p_\sigma$ for $p_\sigma=2,\,\dotsc,\,i_2$ in $8p_\sigma\log p_\sigma+O(p_\sigma)$ time since we can reduce each DFT of order $p_\sigma$ to 2 DFTs of order $2^{n_1}>2p_\sigma-4$, $2^{n_1}<4p_\sigma$. This is possible since the DFT of order $2^{n}>2\cdot\frac{p-1}{2}-4$ is defined over $k$. In the same way, the DFT of order $p_\sigma$ for $\sigma=i_1+1,\,\dotsc,\,i_2$ can be computed in $16p_\sigma\log p_\sigma+O(p_\sigma)$ steps since $2^n\cdot\frac{p-1}{2}>2\cdot\frac{p+1}{2}-4$. Continuing this process we obtain the following upper bound:
\begin{multline*}
D(p^N-1)\le(p^N-1)\sum_{j=-1}^{n-1}\sum_{\sigma=i_j+1}^{i_{j+1}}O\prs{d_\sigma\cdot 2^j\log p_\sigma+d_\sigma}\\
=O(p^N\cdot N\cdot\log\prod_{\sigma=1}^s p_\sigma^{d_\sigma})=O(p^{N}\cdot N^2),
\end{multline*}
which completes the proof.
\end{proof}

\begin{remark}\label{rmk:suit}
For a fixed odd prime $p$, the DFT of order $p^{2^n}-1$ is defined in the field $\FF_{p^{2^n}}$ since the multiplicative group $\FF_{p^{2^n}}^\ast$ of order $p^{2^n}-1$ is cyclic. Corollary~\ref{cor:dft} implies that the DFT of order $p^{2^n}-1$ can be computed in $O(p^{2^n}\cdot 2^{2n})$ steps over $\FF_{p}$. A similar argument shows that the same holds for any field of characteristic $p$ which contains $\FF_{p^{2^k}}$ as a subfield.
\end{remark}

\section{Unified Approach for Fast Polynomial Multiplication}\label{sec:approach}

In this section we present our main contribution. We proceed as follows: first we introduce the notions of the degree function and of the order sequence of a field. Then we describe the DFT-based algorithm $\dD_k$ which computes the product of two polynomials over a field $k$. We show that $\dD_k$ generalizes any algorithm for polynomial multiplication that relies on consecutive applications of DFT, and in particular, Sch\"onhage-Strassen's~\cite{SS71}, Sch\"onhage's~\cite{Sc77}, and Cantor-Kaltofen's~\cite{Caka} algorithms for polynomial multiplication are special cases of the algorithm $\dD_k$. We prove that both the upper and the lower bounds for the total complexity of the algorithm $\dD_k$ depend on the degree function of $k$ and the existence of special order sequences for $k$. In particular, we show that $L_{\dD_k}(n)=\Omega(n\log n)$ when $k$ is a finite field, and $L_{\dD_\QQ}(n)=\Omega(n\log n\log\log n)$. Furthermore, we show sufficient conditions on the field $k$ for the algorithm $\dD_k$ to compute the product of two degree $n$ polynomials in $o(n\log n\log\log n)$, that is, to outperform Sch\"onhage-Strassen's, Sch\"onhage's and Cantor-Kaltofen's algorithms. Finally, we pose a number-theoretic conjecture whose validity would imply faster polynomial multiplication over arbitrary fields of positive characteristic.

In what follows $k$ always stands a field.

\subsection{Extension Degree and Order Sequence}

\begin{definition}
The \emph{degree function} of $k$, is $f_k(n)=[k(\omega_n):k]$ for any positive $n$, where $\omega_n$ is a primitive $n$-th root of unity in the algebraic closure of $k$.
\end{definition}
For example, $f_k(n)=1$ if $k$ is algebraically closed, $f_\RR(n)=1$ if $n\le 2$ and $f_\RR(n)=2$ for $n\ge 3$, $f_\QQ(n)=\phi(n)$ where $\phi(N)$ is as before the Euler's totient function.

An important idea behind F\"urer's algorithm~\cite{Fuer,Deku} is a field extension of small degree containing a principal root of unity of high smooth order. In case of integer multiplication, the characteristic of the ground ring is a parameter we can choose~\cite{Deku}, and it allows us to pick such $\ZZ_{p^c}$ that $p^c-1$ has a large smooth factor. However, in case of multiplication of polynomials over fields, we cannot change the characteristic of the ground field. In what follows we explore this limitation.

\begin{definition}
An integer $n>0$ is called \emph{$c$-suitable} over the field $k$, if the DFT of order $n$ is defined over $k$ and $D_k(n)\le cn\log n$.
\end{definition}

It follows from Corollary~\ref{cor:coltuk} that any $c$-smooth $n$ is $c$-suitable over $k$ as long as the DFT of order $n$ is defined over $k$, and Lemma~\ref{lem:rader} also implies, that if for each prime divisor $p$ of $n$, $p$, or $p-1$ or some $n'\ge 2p-3$, $n'=O(p)$ is $c$-suitable over $k$, then $n$ is $O(c)$-suitable. If $\chr k\ge 3$, then the integers $(\chr k)^{2^n}-1$ are $2^n$-suitable over $k$ for arbitrary $n$ (see Remark~\ref{rmk:suit}).
\begin{definition}
Let $s(n):\NN\to\RR$ be such that $s(n)>1$. A sequence
$$
\nN=\{n_1,\,n_2,\,\dotsc\}
$$
is called an \emph{order sequence} of sparseness $s(n)$ for the field $k$, if
$$
n_i<n_{i+1}\le s(n_i)n_i
$$
and $n_i\mid n_{i+1}$ for $i\ge 1$, and $n_i=n_i'n_i''$, such that there exists a ring extension of $k$ of degree $n_i'$ containing an $n_i''$-th principal root of unity $\omega_{n_i''}$, which is $O(1)$-suitable over this extension. If $s(n)\le C$ for some constant $C$, then $\nN$ is called an order sequence of \emph{constant sparseness}.
\end{definition}

It follows from Remark~\ref{rmk:suit} that $n_i=2^i\cdot(p^{2^i}-1)$ is almost an order sequence of sparseness $s(n)=n$ for any field of characteristic $p$. Decreasing the upper bound for the computation of DFT from $O(n\log^2 n)$ to $O(n\log n)$ would turn it into an order sequence.

\begin{remark}\label{rmk:fields}
If $\chr k\neq 2$, then for the order sequence $\nN=\{2^{i}\}_{i\ge 1}$, $f_k(n'')\le\frac{n''}{2}$ for each $n=n'n''\in\nN$ since if for $n\in\nN$, $\omega_{n''}$, $n''=2^{\cfrc{i-1}{2}}$, $n''=2^{\ffrc{i-1}{2}+1}$ is a primitive $n''$-th root of unity in the algebraic closure of $k$, then
$$
k(\omega_{n''})\cong k[x]/p(x)
$$
and $p(x)\mid x^{\frac{n''}{2}}+1$. The same argument shows that if $\chr k\neq 3$ and
$$
\nN=\{2\cdot3^{i}\}_{i\ge 1},
$$
then $f_k(n'')\le\frac{2n''}{3}$ for each $n=n'n''\in\nN$, $n'=2\cdot 3^{\cfrc{i-1}{2}}$, $n''=3^{\ffrc{i-1}{2}+1}$, since for $k(\omega_{n''})\cong k[x]/p(x)$, $p(x)\mid x^{\frac{2n''}{3}}+x^{\frac{n''}{3}}+1$. Both these order sequences have constant sparsenesses.
\end{remark}

\begin{definition} A field $k$ is called
\begin{itemize}
\item \emph{Fast}, if there is an order sequence $\nN$ of constant sparseness such that $f_k(n_i')=O(1)$ for all $n_i=n_i'n_i''\in\nN$;
\item \emph{$t(n)$-Fast}, if there exists an order sequence $\nN$ of constant sparseness such that $f_k(n_i')\le t(n_i')$ for all $n_=n_i'n_i''\in\nN$.
\item \emph{$t(n)$-Slow}, if for any order sequence $\nN$ of constant sparseness,
$$
f_k(n_i')\ge t(n_i')
$$
for all $n_i=n_i'n_i''\in\nN$.
\end{itemize}
\end{definition}

For example, any algebraically closed field is fast, $\RR$ is a fast field, and $\QQ$ is a $\phi(n)$-slow field, in particular, $\QQ$ is an $\frac{n}{2\log n}$-slow field. It follows from Remark~\ref{rmk:fields}, that any field of characteristic different from $2$ is $\frac{n}{2}$-fast, and any field of characteristic different from $3$ is $\frac{2n}{3}$-fast.

If we want to extend a $b(n)$-slow field $k$ with an $n$-th root of unity, the degree of the extension will be $\Omega(b(n))$. We will see, that to increase performance of a DFT-based algorithm for computing the product of two degree $n-1$ polynomials over $k$, we need to take an extension $K\supseteq k$ of degree $n_1$ over $k$, such that $K$ contains a primitive $n_2$-th root of unity. We will want $n_2$ to be a large suitable number and to belong to a ``not too sparse'' order sequence, preferably of constant sparseness, $n_1$ to be small such that $2n-1\le n_1n_2=O(n)$.

We close this subsection with introducing some technical notation. for a function $f:\NN\to\NN$, such that $\limsup_{n\to\infty}f(n)=\infty$, we will denote by $f^\prt(n)$ the minimal value $f(i)$ over all integer solutions $i$ of the inequality
$$
i\cdot f(i)\ge n.
$$
For example, $n^\prt=\ceil{\sqrt{n}}$, $\sPrs{\frac{n}{\log n}}^\prt\sim\sqrt{\frac{n}{\log n}}$ for $n\ge 2$,\footnote{By $f(n)\sim g(n)$ we denote $f(n)=(1\pm o(1))g(n)$.} and for $q\ge 2$, $(\log_q n)^\prt=\log_q n-\Theta(\log_q\log_q n)$ if $n\ge q$.

We will need to restrict the possible values for $i$ in the inequality to be taken from some order sequence.

For a monotonically growing function $f:\NN\to\NN$, such that $\lim_{n\to\infty}\frac{f(n)}{n}<1$, we will define $f^{(0)}(n)=n$, and for $i\ge 1$, $f^{(i)}(n)=f^{(i-1)}(f(n))$. For each $n\ge 1$, there exists the value $i=i(n)$ such that
$$
f^{(i-1)}(n)\neq f^{(i)}(n)=f^{(i+1)}(n)=\dotsb.
$$
This value will be denoted by $f^\ast(n)$. For example,
\begin{align*}
\prs{\cfrc{n}{2}}^\ast&=\ceil{\log n},&\prs{\ceil{\sqrt{n}}}^\ast&=\ceil{\log\log n},&\prs{\ceil{\log n}}^\ast&=\ceil{\log^\ast n}.
\end{align*}

\subsection{Generalized Algorithm For Polynomial Multiplication}

The DFT-based algorithm $\aA$, the Sch\"onhage-Strassen's and Sch\"onhage's algorithms $\bB$, and the Cantor-Kaltofen's algorithm $\ccC$ are all based on the idea of a field extension with roots of unity of large smooth orders to reduce the polynomial multiplication to many polynomial multiplications of smaller degrees by means of DFT. The natural metaflow of all these algorithms can be generalized as follows: let $\nN$ be an order sequence of constant sparseness over a field $k$, for two polynomials $a(x)$ and $b(x)$ of degree $n-1$ over $k$:
\begin{description}
\item[Embed] Choose a polynomial $P_N(x)$ of degree $N=N'N''\in\nN$,
$$
2n-1\le N=O(n),
$$
and switch to multiplication in $A_N\coloneqq k[x]/P_N(x)$. From this moment consider $a(x)$ and $b(x)$ as elements of $A_N$. There should be an efficiently computable by means of DFTs injective homomorphism $\psi:A_N\to (A_{N'})^{2N''}$, where $A_{N'}\cong k[y]/P_{N'}(y)$ for some $P_{N'}(y)\in k[y]$, and $A_{N'}$ contains a principal $N''$-th (or $2N''$-th) root of unity.
\item[Transform] By means of DFTs over $A_{N'}$ compute
\begin{align*}
\tilde a&\coloneqq\psi(a(x)),&\tilde b&\coloneqq\psi(b(x)),
\end{align*}
both in $(A_{N'})^{2N''}$.
\item[Multiply] Compute $2N''$ products $\tilde c\coloneqq\tilde a\cdot\tilde b$ in $A_{N'}$.
\item[Back-Transform] By means of DFT compute $c(x)=\psi^{-1}(\tilde c)$, which is the ordinary product of the input polynomials.
\item[Unembed] Reduce the product modulo $P_N(x)$ to return the product in $A_N$.
\end{description}

\begin{theorem}
The algorithm $\aA$, Sch\"onhage-Strassen's and Sch\"onhage's algorithms $\bB$ and Cantor-Kaltofen's algorithm $\ccC$ are instances of the algorithm $\dD$.
\end{theorem}
\begin{proof}
For a field $k$ which contains an $N$-th primitive root of unity for
$$
N=2^{\ceil{\log(2n-1)}},
$$
$N=O(n)$, set $P_N(x)=x^N-1$, $N'=1$, $N''=N$ and $A_{N'}=k$. Then $\psi$ is the DFT of order $2N$ (which can be trivially reduced to $N$ in this case) over $k$ and the algorithm $\dD$ appears to be the algorithm $\aA$.

For a field $k$ of characteristic different from $2$, for $\nu=\ceil{\log(2n-1)}$ and $N=2^\nu$, set $P_N(x)=x^N+1$, $N'=2^{\cfrc{\nu}{2}}$, and $N''=2^{\ffrc{\nu}{2}}$. Then $\psi$ is the DFT of order $2N''$ over $A_{N'}$ and the algorithm $\dD$ appears to be the Sch\"onhage-Strassen's algorithm $\bB$~\cite{SS71}.

For $\chr k=2$, set $\nu=\ceil{\log_3(n-\frac{1}{2})}$, $N=3^\nu$, and $P_{2N}(x)=x^{2N}+x^N+1$, $N'=3^{\cfrc{\nu}{2}}$, and $N''=3^{\ffrc{\nu}{2}}$. Then $\psi$ is the DFT of order $3N''$ over $A_{N'}$. However, to fetch the entries of the product in $A_{N'}$ by means $\psi^{-1}$, $2N''$ products of polynomials in $A_{N'}$ are sufficient~\cite{Sc77}. Therefore, the algorithm $\dD$ appears to be the Sch\"onhage's algorithm $\bB$.

For an arbitrary field $k$ fix a positive integer $s\neq\chr p$ and find the least $\nu$ such that $N=\phi(s^\nu)=s^{\nu-1}\phi(s)\ge 2n-1$, and let $\hat N=s^\nu$. Set $P_{\hat N}(x)=\Phi_{\hat N}(x)$, $N'=\phi(s^{\ffrc{\nu}{2}+1})$, and $N''=s^{\cfrc{\nu}{2}-1}$. Then $\psi=\alpha\circ\beta$ where $\alpha$ stands for $2$ DFTs of order $N''$ over $A'$, and $\beta$ is a linear map $A_{N'}[x]\to A_{N'}[x]\times A_{N'}[x]$ such that $\beta(a(x))=(a(x),\,a(\gamma x))$, where $\gamma$ is the $sN''$-th root of unity in $A_{N'}$, i.e., for $A_{N'}\cong k[y]/\Phi_{\hat N'}(y)$, either $\gamma=y$ or $\gamma=y^2$. One can easily show that $\beta$ and $\beta^{-1}$ are computable in linear time. Therefore, the algorithm $\dD$ appears to be the Cantor-Kaltofen's algorithm $\ccC$.
\end{proof}

\subsection{Complexity Analysis}

From the description of the algorithm $\dD$ we have
$$
L_\dD(n)=L'_\dD(N)=2N''L'_\dD(N')+2T(\psi(N))+T(\psi^{-1}(N))
$$
where $L'_\dD(N)$ denotes the complexity of $\dD$ computing the product in $A_N$, $T(\psi(N))$ and $T(\psi^{-1}(N))$ stand for the total complexities of the transformations $\psi$ and $\psi^{-1}$ on inputs of length $N$ respectively.

\begin{theorem}\label{thm:main}
Let the algorithm $\dD$ compute the product of two polynomials in $A_N$ in $\ell$ recursive steps and let $N'=N'_\lambda$ and $N''=N''_\lambda$ be chosen on the step $\lambda=1,\,\dotsc,\,\ell$ \textup($N'_0=N$, $N'_\ell=O(1)$\textup), and for $M(N'_\lambda)=\max\{1,\,\frac{M^\ast(N'_\lambda)}{N'_\lambda}\}$, where $M^\ast(N'_\lambda)$ stands for the complexity of multiplication of an element in $A_{N'_\lambda}$ by powers of an $N''_\lambda$-th root of unity (which exists in $A_{N'_\lambda}$ by assumption). Then
\begin{equation}\label{eq:maineq}
L'_\dD(N)=\Theta\sPRS{N\cdot 2^\ell+N\sum_{\lambda=1}^\ell2^{\lambda-1}\cdot M(N'_\lambda)\log N''_\lambda},
\end{equation}
and if $\chr k\neq 2$, then
\begin{equation}\label{eq:mainlow}
L'_\dD(N)=\Omega\sPRS{N\cdot 2^{(f_k^\prt)^\ast(N)}+N\sum_{\lambda=1}^{(f_k^\prt)^\ast(N)-1}2^{\lambda-1}\log(f_k^\prt)^{(\lambda)}(N)}.
\end{equation}
\end{theorem}
\begin{proof}
Consider the total cost of the algorithm with respect to the computational cost of the first step:
\begin{equation}\label{eq:interm}
L'_\dD(N)=2N''\cdot L'_\dD(N')+\Theta\prs{N''\log N''\cdot(N'+M^\ast(N'))}.
\end{equation}
This follows from the fact that we need to perform a DFT of order $N''$ over $A_{N'}$. Each DFT requires $\Theta(N''\log N'')$ additions of elements in $A_{N'}$ and the same number of multiplications by powers of an $N''$-th principal root of unity. Since $\dim_k A_{N'}=N'$, one addition in $A_{N'}$ takes $N'$ additions in $k$, and by definition, $M^\ast(N')$ is the number of operations in $k$, needed to computed the necessary products by powers of a principal root of unity. Unrolling~\eqref{eq:interm} (by using~\eqref{eq:interm} recursively $\ell$ times),~\eqref{eq:maineq} follows.

To obtain~\eqref{eq:mainlow} from \eqref{eq:maineq} we use the trivial lower bound $M(N')\ge 1$. We then notice that $N'\ge f_k^\prt(N'')$, therefore, we come to the equality $N''_\ell=O(1)$ not earlier than for $\ell=(f_k^\prt)^\ast(N)$, by definition of these operations and the lower bound~\eqref{eq:mainlow} follows.
\end{proof}

\begin{corollary}\label{cor:main}\hfill
\begin{enumerate}
\item For an arbitrary fast field $k$, we have $L_{\dD_k}(n)=O(n\log n)$.
\item For an $o(\log\log n)$-fast field $k$, we have $L_{\dD_k}=o(n\log n\log\log n)$.
\item For an $\Omega(n^{1-o(1)})$-slow field $k$, we have $L_{\dD_k}=\Omega(n\log n\log\log n)$.
\end{enumerate}
\end{corollary}
\proof\hfill
\begin{enumerate}
\item By definition of a fast field, it suffices to take constant number of steps (in fact, even one step) to extend $k$ with a principal root of unity of a suitable order. This means, $\ell=1$ and $N'=O(1)$. Therefore, $M(N')=O(1)$ and trivially $\log N''\le\log N$.
\item By definition of an $o(\log\log n)$-fast field, in the first step we have
$$
N'=o(\log\log N).
$$
We always can bound $M(N'_i)$ with $N'_i$ in~\eqref{eq:maineq}, and we have
$$
\ell=o(\log^\ast\log^\ast n).
$$
Bounding the first summand in the sum in~\eqref{eq:maineq} by
$$
N\cdot N'\cdot\log N=o(n\log n\log\log n),
$$
and each next summand by $o(n\cdot 2^{\log^\ast\log^\ast n}\cdot\log\log n\cdot \log(\log\log n))$, we obtain the statement.
\item For $f_k(n)=\Omega(n^{1-o(1)})$ we have $f_k^\prt(n)=\Omega(n^{\frac{1}{2}-o(1)})$ and
$$
(f_k^\prt)^\ast(n)=\Omega(\log\log n).
$$
Each summand in~\eqref{eq:mainlow} is therefore $\Omega(\log n)$ and the statement follows.\qed
\end{enumerate}

\begin{corollary}
$L_{\dD_\QQ}(n)=\Omega(n\log n\log\log n)$.
\end{corollary}
\begin{proof}
We have $f_\QQ(n)\ge\frac{n}{2\log n}=\Omega(n^{1-o(1)})$ and the statement follows from Corollary~\ref{cor:main}.
\end{proof}

\begin{corollary}
For the finite field $\FF_p$, $L_{\dD_\QQ}(n)=\Omega(n\cdot \log n)$.
\end{corollary}
\begin{proof}
We have $f_{\FF_p}(n)\sim\log_p n$ since the multiplicative group $\FF_p^\ast$ is cyclic and in the extension field $\FF_{p^n}$ of degree $n$ exists a primitive root of unity of order $p^n-1$. This means that $f_{\FF_p}^\prt(n)\sim\log_p n$ and $(f_{\FF_p}^\prt)^\ast(n)\sim\log_p^\ast n$, and the statement follows from taking in~\eqref{eq:mainlow} the first summand which is always $\Theta(n\log n)$.
\end{proof}

Note, that Theorem~\ref{thm:main} does not give any pessimistic lower bound in case of finite fields. Actually, it can give a good upper bound if one can prove existence of order sequences of constant sparseness over finite fields. More formally,

\begin{corollary}
Assume, there exists an order sequence $\nN=\{n_i(p^{n_i}-1)\}_{i\ge 1}$ of constant sparseness over $\FF_p$ and assume that the complexity of multiplication by powers of a principal $(p^{n_i}-1)$-th root of unity in $\FF_{p^{n_i}}$ can be performed in $O(n_i)$ time. Then $L_{\dD_{\FF_p}}(n)=O(n\log n\log^\ast n)$.
\end{corollary}
\begin{proof}
From~\eqref{eq:interm} we get $L'_{\dD_{\FF_p}}(N)\le\frac{2N}{\log_p N}L'_{\dD_{\FF_p}}(\log_p N)+O(N\log N)$, and the statement follows from the solution of this inequality.
\end{proof}

There are two challenges to find a faster polynomial multiplication algorithm over finite fields. The first challenge is the already mentioned existence of order sequences of constant sparseness over these fields. This conjecture is due to Bl\"aser~\cite{Blas}.

\begin{conjecture}[Bl\"aser]
There exist order sequences of constant sparseness over finite fields.
\end{conjecture}
In Remark~\ref{rmk:suit} we showed, that indeed there exist suitable order sequences, however, they are too sparse for our purposes. The second challenge is the complexity of multiplication by powers of a primitive root of unity in extension fields. However, there are ways to overcome this with slight complexity increase. We recently obtained some progress in this area, and we think that a general improvement for fields of characteristic different from $2$ and $0$ is possible.

\section{Conclusion}\label{sec:concl}

We generalized the notion of a DFT-based algorithm for polynomial multiplication, which describes uniformly all currently known fastest algorithms for polynomial multiplication over \emph{arbitrary fields}. We parameterized fields by introducing the notion of the degree function and order sequences and showed upper and lower bounds for DFT-based algorithm in terms of these paremeters.

There is still an important open question whether one can improve the general Sch\"on\-ha\-ge-Stras\-sen's upper bound. As an outcome of this paper we support the general experience that this question is not very easy. In particular, using only known DFT-based techniques will unlikely help much in case of arbitrary fields, in particular for the case of the rational field, as they did for the complexity of integer multiplication.

\subsection*{Acknowledgements}

I would like to thank Markus Bl\"aser for the problem setting and a lot of motivating discussions and to anonymous referees for many important improvement suggestions.


\begin{thebibliography}{99}
\bibitem{Ba03}{S.~Ballet. Low increasing tower of algebraic function fields and bilinear complexity of multiplication in any extension of $\FF_q$. Finite Fields and Their Applications 9, pp.~472--478 (2003).}
\bibitem{BBR09}{S.~Ballet, D.~Le~Brigand, and R.~Rolland. On an application of the definition field descent of a tower of function fields. In Proceedings of the Conference Arithmetic, Geometry and Coding Theory (AGCT 2005), v.~21, pp.~187--203, Soci\'et\'e Meth\'ematique de France, s\'er. S\'eminaires et Congr\'es, 2009.}
\bibitem{BC04}{S.~Ballet and J.~Chaumine. On the bounds of the bilinear complexity of multiplication in some finite fields. Applicable Algebra in Engineering and Computing 15, pp.~205--211 (2004).}
\bibitem{BP10}{S.~Ballet and J.~Pieltant. On the Tensor Rank of Multiplication in Any Extension of $\FF_2$. \href{http://arxiv.org/abs/1003.1864v1}{arXiv:1003:1864v1 [math.AG]} 9 Mar 2010.}
\bibitem{Blas}{M.~Bl\"aser. Private communication.}
\bibitem{BD80}{M.~R.~Brown and D.~P.~Dobkin. An improved lower bound on polynomial multiplication. IEEE Trans.~Comput. 29, pp.~337--340 (1980).}
\bibitem{BK90}{N.~H.~Bshouty and M.~Kaminski. Multiplication of Polynomials over Finite Fields. SIAM J.~Comput. 19(3), pp.~452-456 (1990).}
\bibitem{BK06}{N.~H.~Bshouty and M.~Kaminski. Polynomial multiplication over finite fields: from quadratic to straight-line complexity. Computational Complexity 15(3), pp.~252--262 (2006).}
\bibitem{Burg}{P.~B\"urgisser, M.~Clausen, and A.~ Shokrollahi. Algebraic Complexity Theory. Springer, Berlin, 1997.}
\bibitem{BL04}{P.~B\"urgisser, M.~Lotz. Lower bounds on the bounded coefficient complexity of bilinear maps. J.~ACM 51(3), pp.~464--482 (2004).}
\bibitem{Caka}{D.~G.~Cantor and E.~Kaltofen. On fast multiplication of polynomials over arbitrary algebras. Acta Informatica 28, pp.~693--701 (1991).}
\bibitem{Chud}{D.~Chudnovsky and G.~Chudnovsky. Algebraic complexities and algebraic curves over finite fields. Journal of Complexity 4, pp.~285--316 (1988).}
\bibitem{Clau}{M.~Clausen and U.~Baum. Fast Fourier Transforms. Wissenschaftsverlag, Mannheim-Leipzig-Wien-Z\"urich, 1993.}
\bibitem{Cool}{J.~W.~Cooley and J.~W.~Tukey. An algorithm for the machine calculation of complex Fourier series. Math. Comput. 19, pp.~297-Ð301 (1965).}
\bibitem{Deku}{A.~De, P.~P.~Kurur, C.~Saha, and R.~Saptharishi. Fast integer multiplication using modular arithmetic. In Proceedings of the 40th ACM STOC 2008 conference, pp.~499--506.}
\bibitem{Fuer}{M.~F\"urer. Faster Integer Multiplication. In Proceedings of the 39th ACM STOC 2007 conference, pp.~57--66.}
\bibitem{Ka88}{M.~Kaminski. An algorithm for polynomial multiplication that does not depend on the ring of constants. J.~Algorithms 9, pp.~137--147 (1988).}
\bibitem{Ka05}{M.~Kaminski. A Lower Bound On the Complexity Of Polynomial Multiplication Over Finite Fields. SIAM J.~Comput. 34(4), pp.~960--992 (2005).}
\bibitem{KB89}{M.~Kaminski and N.~H.~Bshouty. Multiplicative Complexity of Polynomial Multiplication over Finite Fields. J.~ACM 36(1), pp.~150--170 (1989).}
\bibitem{HS69}{H.~Hatalov\'a and T. \v{S}al\'at. Remarks on two results in the elementary theory of numbers. Acta Fac. Rer. Natur Univ. Comenian. Math. 20, pp.~113-Ð117 (1969).}
\bibitem{Pan94}{V.~Y.~Pan. Simple Multivariate Polynomial Multiplication. J.~Symbolic Computation 18, pp.~183--186 (1994).}
\bibitem{Rader}{C.~M.~Rader. Discrete Fourier transforms when the number of data samples is prime. Proc. IEEE 56, pp.~1107Ð-1108 (1968).}
\bibitem{Sc77}{A.~Sch\"onhage. Schnelle Multiplikation von Polynomen \"uber K\"orpern der Charakteristic 2. Acta Informatica 7, pp.~395--398 (1977).}
\bibitem{SS71}{A.~Sch\"onhage and V.~Strassen. Schnelle Multiplikation gro\ss er Zahlen. Computing 7, pp.~281--292 (1971).}
\setcounter{bibcnt}{\value{enumiv}}
\bibitem{Shpr}{I.~E.~Shparlinski, M.~A.~Tsfasman, and S.~G.~Vladut. Curves with many points and multiplication in finite fields. Lecture Notes in Math. vol.~1518, Springer-Verlag, Berlin, pp.~145-Ð169 (1992).}
\end{thebibliography}
\end{document}